\documentclass{amsart}
\usepackage{amsfonts}
\usepackage{amsmath}
\usepackage{amssymb}

\theoremstyle{plain}
\newtheorem{theorem}{Theorem}[section]

\newtheorem{proposition}[theorem]{Proposition}
\theoremstyle{definition}
\newtheorem{definition}[theorem]{Definition}

\theoremstyle{remark}
\newtheorem{remark}[theorem]{Remark}
\numberwithin{equation}{section}

\begin{document}
\title[Discrete approximations for the nonlinear Schr\"{o}dinger dynamical
system]{Discrete approximations on functional classes for the integrable
nonlinear Schr\"{o}dinger dynamical system: \\
a symplectic finite-dimensional reduction approach}
\author[Jan L. Cie\'{s}li\'{n}ski and Anatolij Prykarpatski]{Jan L. Cie\'{s}%
li\'{n}ski$^{\ast )}$ and Anatolij K. Prykarpatski$^{\ast \ast )}$}
\address{ $^{\ast )}$ Uniwersytet w Bia\l ymstoku, Wydzia{\l } Fizyki, ul.
Lipowa 41, 15-424 Bia\l ystok, Poland }
\email{janek@alpha.uwb.edu.pl}
\address{$^{\ast \ast )}$ AGH University of Science and Technology, Dept. of
Applied Mathematics, Krak\'ow, Poland }
\email{pryk.anat@ua.fm}

\begin{abstract}
We investigate discretizations of the integrable discrete nonlinear
Schr\"odinger dynamical system and related symplectic structures. We develop
an effective scheme of invariant reducing the corresponding infinite system
of ordinary differential equations to an equivalent finite system of
ordinary differential equations with respect to the evolution parameter. We
construct a finite set of recurrent algebraic regular relations allowing to
generate solutions of the discrete nonlinear Schr\"odinger dynamical system
and we discuss the related functional spaces of solutions. Finally, we
discuss the Fourier transform approach to studying the solution set of the
discrete nonlinear Schr\"odinger dynamical system and its
functional-analytical aspects.
\end{abstract}

\maketitle

\section{Introduction}

With its origins going back several centuries, discrete analysis becomes now
an increasingly central methodology for many mathematical problems related
to discrete systems and algorithms, widely applied in modern science. Our
theme, being related with studying integrable discretizations of nonlinear
integrable dynamical systems and the limiting properties of their solution
sets, is of deep interest in many branches of modern science and technology,
especially in discrete mathematics, numerical analysis, statistics and
probability theory as well as in electrical and electronic engineering. In
fact, this topic belongs to a much more general realm of mathematics, namely
to calculus, differential equations and differential geometry. Thereby,
although the topic is discrete, our approach to treating this problem will
be completely analytical.

In this work we will analyze the properties of discrete approximation for
the nonlinear integrable Schr\"{o}dinger (NLS) dynamical system on a
functional manifold $\tilde{M}\subset L_{2}(\mathbb{R};\mathbb{C}^{2})$:
\begin{equation}
\left.
\begin{array}{l}
\frac{d}{dt}\psi \ =i\psi _{xx}-2i\alpha \psi \psi \psi ^{\ast }, \\[2ex]
\frac{d}{dt}\psi ^{\ast }\ =-i\psi _{xx}^{\ast }+2i\alpha \psi ^{\ast }\psi
\psi ^{\ast }%
\end{array}%
\right\} :=\tilde{K}[\psi ,\psi ^{\ast }],  \label{eq1.1}
\end{equation}%
where, by definition $(\psi ,\psi ^{\ast })^{\intercal }\in \tilde{M},$ $%
\alpha \in \mathbb{R}$ is a constant, the subscript $"x"$ means the partial
derivative with respect to the independent variable $x\in \mathbb{R},$ $%
\tilde{K}:\tilde{M}\rightarrow T(\tilde{M})$ is the corresponding vector
field on $\tilde{M}$ and $t\in \mathbb{R}$ is the evolution parameter. \ The
system (\ref{eq1.1}) possesses  a Lax type representation (see \cite{No})
and is Hamiltonian
\begin{equation}
\frac{d}{dt}(\psi ,\psi ^{\ast })^{\intercal }=-\tilde{\theta}\mathrm{grad}%
\tilde{H}[\psi ,\psi ^{\ast }]=\tilde{K}[\psi ,\psi ^{\ast }]  \label{eq1.1a}
\end{equation}%
with respect to the canonical Poisson structure $\tilde{\theta}$ and the
Hamiltonian function $\tilde H$, where
\begin{equation}
\tilde{\theta}:=\left(
\begin{array}{cc}
0 & -i \\
i & 0%
\end{array}%
\right)  \label{eq1.1b}
\end{equation}%
is a non-degenerate mapping $\tilde{\theta}:T^{\ast }(\tilde{M})\rightarrow
T(\tilde{M})$ on the smooth functional manifold $\tilde{M},$ and
\begin{equation}
\tilde{H}:=\frac{1}{2}\int\limits_{\mathbb{R}}dx\left[ \psi \psi _{xx}^{\ast
}+\psi _{xx}\psi ^{\ast }-2\alpha (\psi ^{\ast }\psi )^{2}\right] ,
\label{eq1.1c}
\end{equation}%
is a smooth mapping $\ \tilde{H}:\tilde{M}\rightarrow \mathbb{C}$. The
corresponding symplectic structure \cite{AM,Ar,Bl,BPS} for the Poissonian
operator \ (\ref{eq1.1b}) is defined by
\begin{eqnarray}
\tilde{\omega}^{(2)} &:&=-\frac{i}{2}\int_{\mathbb{R}}dx[<(d\psi ,d\psi
^{\ast })^{\intercal },\wedge \tilde{\theta}^{-1}(d\psi ,d\psi ^{\ast
})^{\intercal }>=  \label{eq1.1d} \\
&=&-i\int_{\mathbb{R}}dx[d\psi ^{\ast }(x)\wedge d\psi (x)],  \notag
\end{eqnarray}%
which is a non-degenerate and closed 2-form on the functional manifold $%
\tilde{M}.$

The simplest spatial discretizations of the dynamical system (\ref{eq1.1})
look as the flows%
\begin{equation}
\begin{array}{l}
\frac{d}{dt}\psi _{n}=\frac{i}{h^{2}}(\psi _{n+1}-2\psi _{n}+\psi
_{n-1})-2i\alpha \psi _{n}\psi _{n}\psi _{n}^{\ast }, \\[2ex]
\frac{d}{dt}\psi _{n}^{\ast }=-\frac{i}{h^{2}}(\psi _{n+1}^{\ast }-2\psi
_{n}^{\ast }+\psi _{n-1}^{\ast })+2i\alpha \psi _{n}^{\ast }\psi _{n}\psi
_{n}^{\ast }%
\end{array}
\label{eq1.2}
\end{equation}%
and%
\begin{equation}
\left.
\begin{array}{l}
\frac{d}{dt}\psi _{n}=\frac{i}{h^{2}}(\psi _{n+1}-2\psi _{n}+\psi
_{n-1})-i\alpha (\psi _{n+1}+\psi _{n-1})\psi _{n}\psi _{n}^{\ast }, \\[2ex]
\frac{d}{dt}\psi _{n}^{\ast }=-\frac{i}{h^{2}}(\psi _{n+1}^{\ast }-2\psi
_{n}^{\ast }+\psi _{n-1}^{\ast })+i\alpha (\psi _{n+1}^{\ast }+\psi
_{n-1}^{\ast })\psi _{n}\psi _{n}^{\ast },%
\end{array}%
\right\} :=K[\psi _{n},\psi _{n}^{\ast }]  \label{eq1.3}
\end{equation}%
on some "discrete" submanifold $M_{h},$ where, by definition, $\{(\psi _{n},$
$\psi _{n}^{\ast })^{\intercal }\in \mathbb{C}^{2}:$ $n\in \mathbb{%
Z\}\subset }\ $ $M_{h}$ $\subset l_{2}(\mathbb{Z};\mathbb{C}^{2})$ and $%
K:M_{h}\rightarrow T(M_{h})$ is the corresponding vector field on $M_{h}.$

\begin{definition}
If for a function $(\psi ,$ $\psi ^{\ast })^{\intercal }\in W_{2}^{2}(%
\mathbb{R};\mathbb{C}^{2})\ $there exists the point-wise limit $%
\lim_{h\rightarrow 0}(\psi _{n},$ $\psi _{n}^{\ast })^{\intercal }=(\psi (x),
$ $\psi ^{\ast }(x)))^{\intercal },\ $where the set of vectors $(\psi _{n},$
$\psi _{n}^{\ast })^{\intercal }\in \mathbb{C}^{2},n\in \mathbb{Z},$ solves
the infinite system of equations \ (\ref{eq1.3}), the set $\{(\psi _{n},$ $%
\psi _{n}^{\ast })^{\intercal }\in \mathbb{C}^{2}:$ $n\in \mathbb{Z\}\subset
}\ l_{2}(\mathbb{Z};\mathbb{C}^{2})$ will be called an approximate solution
to the nonlinear Schr\"{o}dinger dynamical system \ (\ref{eq1.1}).
\end{definition}

\ It is well known \cite{AL,AL1} that the discretization scheme (\ref{eq1.3}%
) conserves the Lax type integrability \cite{No,Bl,BPS} and that the scheme (%
\ref{eq1.2}) does not. The integrability of \ (\ref{eq1.3}) can be easily
enough checked by means of either the gradient-holonomic integrability
algorithm \cite{Pr,PBPS,BPS} or the well known \cite{LWY} symmetry approach.
In particular, the discrete dynamical system (\ref{eq1.3}) is a Hamiltonian
one \cite{AM,Ar,Bl,Pr} on the symplectic manifold $M_{h}\subset l_{2}(%
\mathbb{Z};\mathbb{C}^{2})$ with respect to the non-canonical symplectic
structure%
\begin{equation}
\omega _{h}^{(2)}=-\sum_{n\in \mathbb{Z}}\frac{ih}{2(1-h^{2}\alpha \psi
_{n}^{\ast }\psi _{n})}<(d\psi _{n},d\psi _{n}^{\ast })^{\intercal },\wedge
(d\psi _{n},d\psi _{n}^{\ast })^{\intercal }>  \label{eq1.4}
\end{equation}%
on $M_{h}$ and looks as
\begin{equation}
\frac{d}{dt}(\psi _{n},\psi _{n}^{\ast })^{\intercal }=-\theta _{n}\mathrm{%
grad}H[\psi _{n},\psi _{n}^{\ast }]=K[\psi _{n},\psi _{n}^{\ast }],
\label{eq1.5}
\end{equation}%
where the Hamiltonian function
\begin{equation}
H=\sum_{n\in \mathbb{Z}}\frac{1}{h}(\psi _{n}\psi _{n+1}^{\ast }+\psi
_{n+1}\psi _{n}^{\ast }+\frac{2}{\alpha h^{2}}\ln |1-\alpha h^{2}\psi
_{n}^{\ast }\psi _{n}|)  \label{eq1.5a}
\end{equation}%
and the related Poissonian operator $\theta _{n}:$ $T^{\ast
}(M_{h})\rightarrow T(M_{h})$ equals
\begin{equation}
\theta _{n}:=\left(
\begin{array}{cc}
0 & -ih^{-1}(1-h^{2}\alpha \psi _{n}^{\ast }\psi _{n}) \\
ih^{-1}(1-h^{2}\alpha \psi _{n}^{\ast }\psi _{n}) & 0%
\end{array}%
\right) .  \label{eq1.6}
\end{equation}

\begin{remark}
\label{Rm_1.1}For the symplectic structure \ (\ref{eq1.4}) and,
respectively, the Hamiltonian function \ (\ref{eq1.5a}) to be suitably
defined on the manifold $M_{h}\subset l_{2}(\mathbb{Z};\mathbb{C}^{2})$ it
is necessary to assume additionally that the finite stability condition $%
\lim_{N,M\rightarrow \infty }\left( \prod\limits_{-N}^{M}\ (1-\alpha
h^{2}\psi _{n}^{\ast }\psi _{n})\right) \neq 0$ holds. The latter is simply
reduced as $h\rightarrow 0$ to the equivalent integral inequality $\alpha
\leq \int_{\mathbb{R}}(x\psi ^{\ast }\psi )^{2}dx\ \left( \int_{\mathbb{R}%
}\psi ^{\ast }\psi dx\right) ^{-1},$ which will be assumed for further to be
satisfied. respectively, the manifold $\tilde{M}\subset \tilde{W}_{2}^{2}(%
\mathbb{R};\mathbb{C}^{2}),$ where $\tilde{W}_{2}^{2}(\mathbb{R};\mathbb{C}%
^{2}):=W_{2}^{2}(\mathbb{R};\mathbb{C}^{2})\cap L_{2}^{(1)}(\mathbb{R};%
\mathbb{C}^{2})$ with the space $L_{2}^{(1)}(\mathbb{R};\mathbb{C}%
^{2}):=\{(\psi ,\psi ^{\ast })^{\intercal }\in L_{2}(\mathbb{R};\mathbb{C}%
^{2}):\int_{\mathbb{R}}x^{2}(\psi ^{\ast }\psi )^{2}dx<\infty \}.$
\end{remark}

The symplectic structure (\ref{eq1.4}) is well defined on the manifold $%
M_{h} $ and tends as $h\rightarrow 0$ to the symplectic structure (\ref%
{eq1.1d}) on $\tilde{M},$ and respectively the Hamiltonian function (\ref%
{eq1.5a}) tends to (\ref{eq1.1c}).

In this work we have investigated the structure of the solution set to the
discrete nonlinear Schr\"{o}dinger dynamical system \ (\ref{eq1.3}) by means
of a specially devised analytical approach for invariant reducing the
infinite system of ordinary differential equations (\ref{eq1.3}) to an
equivalent finite one of ordinary differential equations with respect to the
evolution parameter $t\in \mathbb{R}.$ \ As a result, there was constructed
a finite set of recurrent algebraic regular relationships, allowing to
expand the obtained before finite set of solutions to any discrete order $%
n\in \mathbb{Z},\ $ which makes it possible to present a wide class of the
approximate solutions to the nonlinear Schr\"{o}dinger dynamical system \ (%
\ref{eq1.1}).$\ $ It is worthy here to stress that the problem of
constructing an effective discretization scheme for the nonlinear Schr\"{o}%
dinger dynamical system \ (\ref{eq1.1}) and its generalizations proves to be
important both for applications \cite{APT,Ke,Sa,YJ} and for deeper
understanding the nature of the related algebro-geometric and analytic
structures responsible for their limiting stability and convergence
properties. From these points of view we would like to mention work \ \cite%
{LLC}, where \ the standard discrete Lie-algebraic approach \cite{Bl,BSP}
was recently applied\ to constructing a slightly different from \ (\ref%
{eq1.2}) and \ (\ref{eq1.3}) discretization of the nonlinear Schr\"{o}dinger
dynamical system \ (\ref{eq1.1}). As the symplectic reduction method,
devised in the present work for studying the solution sets of the discrete
nonlinear Schr\"{o}dinger dynamical system \ (\ref{eq1.3}), is completely
independent of a chosen discretization scheme, it would be reasonable and
interesting to apply it to that of \cite{LLC} and compare the corresponding
results subject to their computational effectiveness.

\section{A class of Hamiltonian discretizations of the NLS dynamical system}

The discretizations \ (\ref{eq1.2}) and \ (\ref{eq1.3}) can be extended to a
wide classs of Hamiltonian systems, if to assume that the Poissonian
structure is given by the local expression
\begin{equation}
\theta _{n}=\left(
\begin{array}{cc}
0 & -i\nu _{n}(g_{n}-{\tilde{h}}_{n}^{2}\alpha \psi _{n}^{\ast }\psi _{n})
\\
i\nu _{n}(g_{n}-{\tilde{h}}_{n}^{2}\alpha \psi _{n}^{\ast }\psi _{n}) & 0%
\end{array}%
\right) ,  \label{eq1.6a}
\end{equation}%
generalizing \ (\ref{eq1.6}), and the Hamiltonian function is chosen in the
form%
\begin{equation}
H=\sum_{n\in \mathbb{Z}}h_{n}\left( a_{n}\psi _{n}\psi _{n+1}^{\ast
}+b_{n}\psi _{n}\psi _{n}^{\ast }+c_{n}\psi _{n}\psi _{n-1}^{\ast }+\frac{%
2d_{n}}{\alpha }\ln |g_{n}-\alpha {\tilde{h}}_{n}^{2}\psi _{n}\psi
_{n}^{\ast }|\right) ,  \label{eq1.6b}
\end{equation}%
where $h_{n},{\tilde{h}}_{n},\nu _{n},a_{n},b_{n}$ $,c_{n},d_{n}$ and $g_{n}$
$\in \mathbb{R}_{+},n\in \mathbb{Z},$ are some parameters. The reality
condition, imposed on the Hamiltonian function \ (\ref{eq1.6b}), yields the
relationships%
\begin{equation}
c_{n}h_{n}=a_{n-1}^{\ast }h_{n-1}\ ,\quad b_{n}^{\ast }=b_{n}\ ,\quad
d_{n}^{\ast }=d_{n},  \label{eq1.6c}
\end{equation}
which should be satisfied for all $n\in \mathbb{Z}.$ As a result, there is
obtained the corresponding generalized discrete nonlinear Schr\"odinger
dynamical system $\frac{d}{dt}(\psi _{n},\psi _{n}^{\ast })^{\intercal
}:=-\theta _{n}\mathrm{grad}$ $H[\psi _{n},\psi _{n}^{\ast }],$ $n\in
\mathbb{Z},$ equivalent to the infinite set of ordinary differential
equations

\begin{equation}
\begin{array}{l}
\frac{d}{dt}\psi _{n}=i\nu _{n}\left( h_{n+1}c_{n+1}g_{n}\psi
_{n+1}+(b_{n}g_{n}h_{n}-2{\tilde{h}}_{n}^{2}h_{n}d_{n})\psi
_{n}+h_{n-1}a_{n-1}g_{n}\psi _{n-1}\right) - \\
\qquad \quad -i\alpha \nu _{n}{\tilde{h}}_{n}^{2}\left( h_{n+1}c_{n+1}\psi
_{n+1}+h_{n}b_{n}\psi _{n}+h_{n-1}a_{n-1}\psi _{n-1}\right) \psi _{n}\psi
_{n}^{\ast } \\
\frac{d}{dt}\psi _{n}^{\ast }=-i\nu _{n}\left( h_{n}a_{n}g_{n}\psi
_{n+1}^{\ast }+(b_{n}g_{n}h_{n}-2{\tilde{h}}_{n}^{2}h_{n}d_{n})\psi
_{n}^{\ast }+h_{n}c_{n}g_{n}\psi _{n-1}^{\ast }\right) + \\
\qquad \quad +i\alpha \nu_{n}{\tilde{h}}_{n}^{2}\left( h_{n}a_{n}\psi
_{n+1}^{\ast }+h_{n}b_{n}\psi _{n}^{\ast }+h_{n}c_{n}\psi _{n-1}^{\ast
}\right) \psi _{n}\psi _{n}^{\ast }%
\end{array}
\label{eq1.6d}
\end{equation}%
for all $n\in \mathbb{Z}.$ In the completely autonomous case, when $h_{n}=h,{%
\tilde{h}}_{n}=\tilde{h},\nu _{n}=\nu ,a_{n}=a,,b_{n}=b,c_{n}=c,d_{n}=d$ and
$g_{n}=g$ $\in \mathbb{R}_{+}$ \ for all $n\in \mathbb{Z},$ the Poissonian
structure \ (\ref{eq1.6a}) becomes

\begin{equation}
\theta _{n}=\left(
\begin{array}{cc}
0 & -i\nu (g-{\tilde{h}}^{2}\alpha \psi _{n}^{\ast }\psi _{n}) \\
i\nu (g-{\tilde{h}}^{2}\alpha \psi _{n}^{\ast }\psi _{n}) & 0%
\end{array}%
\right)  \label{eq1.6e}
\end{equation}%
and \ the Hamiltonian function \ (\ref{eq1.6b}) becomes

\begin{equation}
H=\sum_{n\in \mathbb{Z}}h\left( a\psi _{n}\psi _{n+1}^{\ast }+b\psi _{n}\psi
_{n}^{\ast }+c\psi _{n}\psi _{n-1}^{\ast }+\frac{2d}{\alpha }\ln |g-\alpha {%
\tilde{h}}^{2}\psi _{n}\psi _{n}^{\ast }|\right) .  \label{eq1.6f}
\end{equation}%
The corresponding reality condition for \ (\ref{eq1.6f}) reads as

\begin{equation}
c=a^{\ast }\ ,\quad b^{\ast }=b\ ,\quad d^{\ast }=d\ ,  \label{eq1.6g}
\end{equation}%
and the related discrete nonlinear Schr\"odinger dynamical systems reads as
a set of the equations

\begin{equation}
\begin{array}{l}
\ \frac{d}{dt}\psi _{n}=i\nu h\left( cg\psi _{n+1}+(bg-2{\tilde{h}}%
^{2}d)\psi _{n}+ag\psi _{n-1}\right) -i\alpha \nu h{\tilde{h}}^{2}\left(
c\psi _{n+1}+b\psi _{n}+a\psi _{n-1}\right) \psi _{n}\psi _{n}^{\ast }, \\%
[2ex]
\frac{d}{dt}\psi _{n}^{\ast }=-i\nu h\left( ag\psi _{n+1}^{\ast }+(bg-2{%
\tilde{h}}^{2}d)\psi _{n}^{\ast }+cg\psi _{n-1}^{\ast }\right) +i\alpha \nu h%
{\tilde{h}}^{2}\left( a\psi _{n+1}^{\ast }+b\psi _{n}^{\ast }+c\psi
_{n-1}^{\ast }\right) \psi _{n}\psi _{n}^{\ast },%
\end{array}
\label{eq1.6h}
\end{equation}%
for all $n\in \mathbb{Z}.$\texttt{\ }If now to make in \ (\ref{eq1.6d}) the
substitutions

\begin{eqnarray}
\nu _{n} &=&\frac{1}{h_{n}}\ ,\quad g_{n}=1\ ,\quad {\tilde{h}}_{n}=h_{n}\
,\quad a_{n}=\frac{1}{h_{n}^{2}}\ ,  \label{eq1.6i} \\
\quad b_{n} &=&0\ ,\quad c_{n}=\frac{1}{h_{n}h_{n-1}}\ ,\quad d_{n}=\frac{1}{%
h_{n}^{4}}\ ,  \notag
\end{eqnarray}%
one \ obtains the discrete nonlinear Schr\"odinger dynamical system%
\begin{equation}
\left.
\begin{array}{c}
\frac{d}{dt}\psi _{n}=\frac{i}{h_{n}^{2}}(\psi _{n+1}-2\psi
_{n}+h_{n}h_{n-1}^{-1}\psi _{n-1})- \\
-i\alpha (\psi _{n+1}+h_{n}h_{n-1}^{-1}\psi _{n-1})\psi _{n}^{\ast }\psi
_{n}, \\
\frac{d}{dt}\psi _{n}^{\ast }=-\frac{i}{h_{n}^{2}}(\psi _{n+1}^{\ast }-2\psi
_{n}^{\ast }+h_{n}h_{n-1}^{-1}\psi _{n-1}^{\ast })+ \\
+i\alpha (\psi _{n+1}^{\ast }+h_{n}h_{n-1}^{-1}\psi _{n-1}^{\ast })\psi
_{n}^{\ast }\psi _{n},%
\end{array}%
\right\} :=K_{n}^{(g)}[\psi _{n},\psi _{n}^{\ast }],  \label{eq1.6j}
\end{equation}
whose Hamiltonian function equals
\begin{equation}
H^{(g)}=\ \sum_{n\in \mathbb{Z}}h_{n}^{-1}(\psi _{n}\psi _{n+1}^{\ast }+\psi
_{n+1}\psi _{n}^{\ast }+\frac{2}{\alpha h_{n}^{2}}\ln |1-\alpha
h_{n}^{2}\psi _{n}^{\ast }\psi _{n}|).  \label{eq1.6k}
\end{equation}
Another substitution, \texttt{\ }taken in the form

\begin{equation}
c=a\neq 0\ ,\quad \nu hga=\frac{1}{h^{2}}\ ,\quad (bg-2{\tilde{h}}^{2}d)\nu
h=-\frac{2}{h^{2}}\ ,\quad \nu h{\tilde{h}}^{2}(a+b+c)=2,\   \label{eq1.6l}
\end{equation}
is also suitable in the limit $h\rightarrow 0$ \ for discretization the
nonlinear Schr\"odinger dynamical system \ (\ref{eq1.3}). The corresponding
discrete nonlinear Schr\"odinger dynamics takes the form
\begin{equation}
\begin{array}{l}
\frac{d}{dt}\psi _{n}=\frac{i}{h^{2}}(\psi _{n+1}-2\psi _{n}+\psi _{n-1})-%
\frac{2i\alpha }{2+\mu }(\psi _{n+1}+\psi _{n-1}+\mu \psi _{n})\psi _{n}\psi
_{n}^{\ast }, \\[2ex]
\frac{d}{dt}\psi _{n}^{\ast }=-\frac{i}{h^{2}}(\psi _{n+1}^{\ast }-2\psi
_{n}^{\ast }+\psi _{n-1}^{\ast })+\frac{2i\alpha }{2+\mu }(\psi _{n+1}^{\ast
}+\psi _{n-1}^{\ast }+\mu \psi _{n}^{\ast })\psi _{n}\psi _{n}^{\ast }%
\end{array}
\label{eq1.6m}
\end{equation}%
for all for all $n\in \mathbb{Z},$ where $\mu =b/a\in \mathbb{R}_{+}.$ \ \
Thus we obtained a one-parameter family of Hamiltonian discretizations of
the NLS equation. The set of relationships (\ref{eq1.6l}) admits a lot of
reductions, for instance, one can take
\begin{equation}
\nu =1,\quad g=1,\quad a=\frac{1}{h^{3}},\quad d=\left( \frac{\mu +2}{2}%
\right) ^{2}\frac{1}{h^{5}},\quad \frac{{\tilde{h}}^{2}}{h^{2}}=\frac{2}{%
2+\mu },  \label{eq1.6n}
\end{equation}%
not changing the infinite set of equations \ (\ref{eq1.6m}).

All of the constructed above discretizations of the nonlinear Schr\"odinger
dynamical system \ (\ref{eq1.1}) on the functional manifold $\tilde{M}$ can
be considered as either better or worse from the computational point of
view. If some of the discretization allows, except the Hamiltonian function,
some extra conservation laws, it can be naturally considered as a much more
suitable for numerical analysis case, allowing both to control the stability
of the solution convergence, as a parameter $\mathbb{R}_{+}\ni h\rightarrow
0,$ and to make an invariant solution space reduction\ to a lower effective
dimension of the related solution set.

It is worthy to observe here that the functional structure of the
discretization (\ref{eq1.3}) strongly depends both on the manifold $M$ and
on the convergent as $h\rightarrow 0\ $ \ form of the Hamiltonian function (%
\ref{eq1.6}). In particular, the existence of the limit
\begin{equation}
\tilde{H}:=\lim\limits_{h\rightarrow 0}\ \sum_{n\in \mathbb{Z}}\frac{1}{h}%
(\psi _{n}\psi _{n+1}^{\ast }+\psi _{n+1}\psi _{n}^{\ast }+\frac{2}{\alpha
h^{2}}\ln |1-\alpha h^{2}\psi _{n}^{\ast }\psi _{n}|),  \label{eq1.8}
\end{equation}%
coinciding with the expression \ (\ref{eq1.1c}), imposes a strong constraint
on the functional space $\tilde{M}\subset L_{2}(\mathbb{R};\mathbb{C}^{2}),$
namely, a vector-function $(\psi ,\psi ^{\ast })^{\intercal }\in W_{2}^{2}(%
\mathbb{R};\mathbb{C}^{2})$ $\subset L_{2}(\mathbb{R};\mathbb{C}^{2}),$
thereby fixing a suitable functional class \cite{AH} for which the
discretization conserves its physical Hamiltonian system sense.
Respectively, the limiting for (\ref{eq1.8}) symplectic structure
\begin{eqnarray}
\tilde{\omega}^{(2)} &:&=-\lim\limits_{h\rightarrow 0}\sum_{n\in \mathbb{Z}}%
\frac{i}{2}<(d\psi _{n},d\psi _{n}^{\ast })^{\intercal },\wedge \theta
_{n}^{-1}(d\psi _{n},d\psi _{n}^{\ast })^{\intercal }>=  \label{eq1.8a} \\
&=&-\lim\limits_{h\rightarrow 0}i\sum_{n\in \mathbb{Z}}h(1-\alpha h^{2}\psi
_{n}^{\ast }\psi _{n})^{-1}d\psi _{n}^{\ast }\wedge d\psi _{n}=-i\int_{%
\mathbb{R}}dx[d\psi ^{\ast }(x)\wedge d\psi (x)]  \notag
\end{eqnarray}%
on the manifold $\tilde{M}$ \ coincides exactly with the canonical
symplectic structure \ (\ref{eq1.1d}) for the dynamical system \ (\ref%
{eq1.1a}).

If now to assume that a vector function $(\psi ,\psi ^{\ast })^{\intercal
}\in W_{2}^{1}(\mathbb{R};\mathbb{C}^{2})\subset L_{2}(\mathbb{R};\mathbb{C}%
^{2}),$ the Hamiltonian function (\ref{eq1.6}) can be taken only as
\begin{equation}
H^{(s)}=\sum_{n\in \mathbb{Z}}(\psi _{n}\psi _{n+1}^{\ast }+\psi _{n+1}\psi
_{n}^{\ast }+\frac{2}{\alpha h^{2}}\ln |1-\alpha h^{2}\psi _{n}^{\ast }\psi
_{n}|),  \label{eq1.9}
\end{equation}%
and the corresponding Poissonian structure as
\begin{equation}
\theta _{n}^{(s)}:=\left(
\begin{array}{cc}
0 & ih^{-2} (h^2 \alpha \psi _{n}^{\ast }\psi -1) \\
ih^{-2}(1-h^{2}\alpha \psi _{n}^{\ast }\psi ) & 0%
\end{array}%
\right)  \label{eq1.10}
\end{equation}%
The limiting for (\ref{eq1.9}) Hamiltonian function
\begin{equation}
\tilde{H}^{(s)}:=\lim\limits_{h\rightarrow 0}H^{(s)}\ =\int\limits_{\mathbb{R%
}}dx(\psi \psi _{x}^{\ast }+\psi _{x}\psi ^{\ast })=0\   \label{eq1.10aa}
\end{equation}%
\ becomes trivial and, simultaneously, the limiting for (\ref{eq1.10})
symplectic structure
\begin{eqnarray}
\tilde{\omega}_{(s)}^{(2)} &:&=\lim\limits_{h\rightarrow 0}\sum_{n\in
\mathbb{Z}}\frac{i}{2}<(d\psi _{n},d\psi _{n}^{\ast })^{\intercal },\wedge
\theta _{n}^{(s),-1}(d\psi _{n},d\psi _{n}^{\ast })^{\intercal }>=  \notag \\
&=&\lim\limits_{h\rightarrow 0}i\sum_{n\in \mathbb{Z}}h^{2}(1-\alpha
h^{2}\psi _{n}^{\ast }\psi _{n})^{-1}d\psi _{n}^{\ast }\wedge d\psi _{n}=0
\label{eq1.10a}
\end{eqnarray}%
becomes trivial too. Thus, the functional space $W_{2}^{1}(\mathbb{R};%
\mathbb{C}^{2})\subset L_{2}(\mathbb{R};\mathbb{C}^{2})$ is not suitable for
the discretization \ (\ref{eq1.3}) of the nonlinear integrable Schr\"odinger
dynamical system \ (\ref{eq1.1}).

It is important here to stress that the discretization parameter $h\in
\mathbb{R}_{+}$ can be taken as depending on the node $n\in \mathbb{Z}:$ $%
h\rightarrow h_{n}\in \mathbb{R}_{+},$ which satisfies the condition $%
\sup\limits_{n\in \mathbb{Z}}h_{n}\ \leq \varepsilon ,$ where the condition $%
\varepsilon \rightarrow 0$ should be later imposed. For instance, one can
replace the dynamical system (\ref{eq1.3}) by\ (\ref{eq1.6j}), the
Poissonian structure (\ref{eq1.4}) by
\begin{equation}
\theta _{n}^{(g)}:=\left(
\begin{array}{cc}
0 & ih_{n}^{-1}(h_{n}^{2}\alpha \psi _{n}^{\ast }\psi -1) \\
ih_{n}^{-1}(1-h_{n}^{2}\alpha \psi _{n}^{\ast }\psi ) & 0%
\end{array}%
\right)   \label{eq1.12}
\end{equation}%
and, respectively, the Hamiltonian function (\ref{eq1.6}) becomes exactly\ (%
\ref{eq1.6k}).

It is easy to check that the modified discrete dynamical system (\ref{eq1.6j}%
) can be equivalently rewritten as
\begin{equation}
\frac{d}{dt}(\psi _{n},\psi _{n}^{\ast })^{\intercal }=-\theta _{n}^{(g)}%
\mathrm{grad}H^{(g)}[\psi _{n},\psi _{n}^{\ast }]\   \label{eq1.14}
\end{equation}%
for all $n\in \mathbb{Z},$ meaning, in particular, that the Hamiltonian
function (\ref{eq1.6k}) is conservative. The latter follows from the fact
that the skewsymmetric operator (\ref{eq1.12}) is Poissonian on the
discretized manifold $M_{h}.\ $\ Moreover, if to impose the constraint that
uniformly in $n\in \mathbb{Z}$ the limit \textrm{lim}$_{\varepsilon
\rightarrow 0}(h_{n}h_{n-1}^{-1})=1,$ the dynamical system \ (\ref{eq1.6j})
reduces to (\ref{eq1.1}) and the corresponding limiting symplectic structure
\begin{eqnarray}
\tilde{\omega}_{(g)}^{(2)} &:&=-\lim\limits_{\varepsilon \rightarrow
0}\sum_{n\in \mathbb{Z}}\frac{i}{2}<(d\psi _{n},d\psi _{n}^{\ast
})^{\intercal },\wedge \theta _{n}^{(g),-1}(d\psi _{n},d\psi _{n}^{\ast
})^{\intercal }>=  \label{eq1.15a} \\
&=&-\lim\limits_{\varepsilon \rightarrow 0}i\sum_{n\in \mathbb{Z}%
}h_{n}(1-\alpha h_{n}^{2}\psi _{n}^{\ast }\psi _{n})^{-1}d\psi _{n}^{\ast
}\wedge d\psi _{n}=  \notag \\
&=&-i\int_{\mathbb{R}}dx[d\psi ^{\ast }(x)\wedge d\psi (x)],  \notag
\end{eqnarray}%
coincides exactly with the symplectic structure (\ref{eq1.8a}).

\begin{remark}
It is, by now, a not solved, but interesting, problem whether the modified
discrete Hamiltonian dynamical system (\ref{eq1.6j}) sustains to be Lax type
integrable. It is left for studying in a separate work.
\end{remark}

\section{Conservation laws for the integrable discrete NLS system}

Taking into account that the discrete dynamical system (\ref{eq1.3}) is well
posed in the space $M_{h}:=w_{h,2}^{2}(\mathbb{Z};\mathbb{C}^{2})\subset
l_{2}(\mathbb{Z};\mathbb{C}^{2}),$ suitably approximating the Sobolev space
of functions $W_{2}^{2}(\mathbb{R};\mathbb{C}^{2}),$ we can go further and
to approximate the space $w_{h,2}^{2}(\mathbb{Z};\mathbb{C}^{2})$ by means
of an infinite hierarchy of strictly invariant finite dimensional subspaces $%
M_{h}^{(N)}\simeq \bar{w}_{h,2}^{2}(\mathbb{Z}_{(N)};\mathbb{C}^{2}),$ $N\in
\mathbb{Z}_{+}.$ In particular, as it was before shown both in \cite{AL,AL1}
by means of the inverse scattering transform method \cite{AL,No} and in \cite%
{BP,Pr,PBPS} \- by means of the gradient-holonomic approach \cite{PM}, the
discrete nonlinear Schr\"{o}dinger dynamical system (\ref{eq1.3}) possesses
on the manifold $M_{h}$ an infinite hierarchy of the functionally
independent conservation laws:
\begin{eqnarray}
&&\bar{\gamma}_{0}=\frac{1}{\alpha h^{3}}\sum_{n\in \mathbb{Z}}\ln |1-\alpha
h^{2}\psi _{n}^{\ast }\psi _{n}|,\ \ \ \ \ \gamma _{0}=\sum_{n\in \mathbb{Z}%
_{+}}\sigma _{n}^{{(0)}},  \label{eq1.15} \\
&&\gamma _{1}=\sum_{n\in \mathbb{Z}}(\sigma _{n}^{(1)}+\frac{1}{2}\sigma
_{n}^{(0)}\sigma _{n}^{(0)}),  \notag \\
&&\gamma _{2}=\sum_{n\in \mathbb{Z}}(\sigma _{n}^{(2)}+\frac{1}{3}\sigma
_{n}^{(0)}\sigma _{n}^{(0)}\sigma _{n}^{(0)}+\sigma _{n}^{(0)}\sigma
_{n}^{(1)}),\ \ ...,  \notag
\end{eqnarray}%
where the quantities $\sigma _{n}^{(j)},$ $n\in \mathbb{Z},$ $j\in \mathbb{Z}%
_{+},$ are defined as follows:
\begin{eqnarray}
&&\sigma _{n}^{(0)}=-\frac{1}{\alpha h^{2}}(\psi _{n}^{\ast }\psi
_{n-1}+\psi _{n-1}^{\ast }\psi _{n-2}),  \label{eq1.16} \\
&&\sigma _{n}^{(1)}=i\frac{d}{dt}\sigma _{n-1}^{(0)}+(1-\alpha h^{2}\psi
_{n-1}^{\ast }\psi _{n-1})(1-\alpha h^{2}\psi _{n-2}^{\ast }\psi _{n-2})+
\beta \frac{\alpha }{h^{2}}\psi _{n-1}^{\ast }(\psi _{n}+\psi _{n-1}),\ ...\
,  \notag
\end{eqnarray}%
and $\beta \in \mathbb{R}$ is an arbitrary constant parameter. As a result
of \ (\ref{eq1.16}) one finds the following infinite hierarchy of smooth
conservation laws:%
\begin{eqnarray}
\bar{H}_{0} &=&\sum_{n\in \mathbb{Z}}\ln |1-\alpha h^{2}\psi _{n}^{\ast
}\psi _{n}|,\   \label{eq1.16a} \\
H_{0} &=&\ \sum_{n\in \mathbb{Z}}\psi _{n}^{\ast }\psi _{n+1},\text{ \ \ }%
H_{0}^{\ast }=\ \sum_{n\in \mathbb{Z}}\psi _{n}\psi _{n+1}^{\ast },  \notag
\\
H_{1} &=&\sum_{n\in \mathbb{Z}}(\frac{1}{2}\psi _{n}^{2}\psi _{n-1}^{\ast
,2}+\psi _{n}\psi _{n+1}\psi _{n-1}^{\ast }\psi _{n}^{\ast }-\frac{\psi
_{n}\psi _{n-2}^{\ast }}{\alpha h^{2}}),  \notag \\
H_{1}^{\ast } &=&\sum_{n\in \mathbb{Z}}(\frac{1}{2}\psi _{n-1}^{2}\psi
_{n}^{\ast ,2}+\psi _{n-1}\psi _{n}\psi _{n+1}^{\ast }\psi _{n}^{\ast }-%
\frac{\psi _{n-2}\psi _{n}^{\ast }}{\alpha h^{2}}),  \notag
\end{eqnarray}%
\begin{eqnarray*}
H_{2} &=&\sum_{n\in \mathbb{Z}}[\frac{1}{3}\psi _{n}^{3}\psi _{n-1}^{\ast
,3}+\psi _{n}\psi _{n+1}\psi _{n-1}^{\ast }\psi _{n}^{\ast }(\psi _{n}\psi
_{n-1}^{\ast }+\psi _{n+1}\psi _{n}^{\ast }+ \\
&&+\psi _{n+2}\psi _{n+1}^{\ast })-\frac{\psi _{n}\psi _{n-1}^{\ast }}{%
\alpha h^{2}}(\psi _{n}\psi _{n-2}^{\ast }+\psi _{n+1}\psi _{n-1}^{\ast })-
\\
&&-\frac{\psi _{n}\psi _{n}^{\ast }}{\alpha h^{2}}(\psi _{n+1}\psi
_{n-2}^{\ast }+\psi _{n+2}\psi _{n-1}^{\ast })+\frac{\ \psi _{n}\psi
_{n-3}^{\ast }}{\alpha ^{2}h^{4}}], \\
H_{2}^{\ast } &=&\sum_{n\in \mathbb{Z}}[\frac{1}{3}\psi _{n}^{\ast ,3}\psi
_{n-1}^{3}+\psi _{n}^{\ast }\psi _{n+1}^{\ast }\psi _{n-1}\psi _{n}(\psi
_{n}^{\ast }\psi _{n-1}+\psi _{n+1}^{\ast }\psi _{n}+  \notag \\
&&+\psi _{n+2}^{\ast }\psi _{n+1})-\frac{\psi _{n}^{\ast }\psi _{n-1}}{%
\alpha h^{2}}(\psi _{n}^{\ast }\psi _{n-2}+\psi _{n+1}^{\ast }\psi _{n-1})-
\\
&&-\frac{\psi _{n}\psi _{n}^{\ast }}{\alpha h^{2}}(\psi _{n+1}^{\ast }\psi
_{n-2}+\psi _{n+2}^{\ast }\psi _{n-1})+\frac{\ \psi _{n}^{\ast }\psi _{n-3}}{%
\alpha ^{2}h^{4}}],
\end{eqnarray*}%
and so on.

Taking into account the functional structure of the equations \ (\ref{eq1.2}%
) or \ (\ref{eq1.3}), one can define the space $\mathcal{D}(M_{h})$ of
smooth functions $\gamma :M_{h}\rightarrow \mathbb{C}$ \ on $M_{h}$ as that
\ invariant with respect to the phase transformation \ $\mathbb{C}^{2}\ni
(\psi _{n},\psi _{n}^{\ast })\rightarrow (e^{\alpha }\psi _{n},e^{-\alpha
}\psi _{n}^{\ast })\in \mathbb{C}^{2}$ for any $n\in \mathbb{Z}$ and $\alpha
\in \mathbb{C}.$ Equivalently, a function $\gamma \in \mathcal{D}(M_{h})$
iff the following condition
\begin{equation}
\sum_{n\in \mathbb{Z}}<\mathrm{grad}\gamma \lbrack \psi _{n},\psi _{n}^{\ast
}],(\psi _{n},-\psi _{n}^{\ast })^{\intercal }>=0  \label{eq1.3a}
\end{equation}%
holds on $M_{h}$. Note that conserved quantities (\ref{eq1.16a}) belong to $%
\mathcal{D}(M_{h})$.

The conservation law $\bar{H}_{0}\in \mathcal{D}(M_{h})$ is a Casimir
function for the Poissonian structure \ (\ref{eq1.4}) on the manifold $M_{h},
$ that is for any $\gamma \in $ $\mathcal{D}(M_{h})$ the Poisson bracket
\begin{eqnarray}
\{\gamma ,\bar{H}_{0}\} &:&=\sum_{n\in \mathbb{Z}}<\mathrm{grad}\gamma
\lbrack \psi _{n},\psi _{n}^{\ast }],\theta _{n}\mathrm{grad}\bar{H}%
_{0}[\psi _{n},\psi _{n}^{\ast }]>=  \notag \\
&=& i \alpha h \sum_{n\in \mathbb{Z}}<\mathrm{grad}\gamma \lbrack \psi
_{n},\psi _{n}^{\ast }],(\psi _{n},-\psi _{n}^{\ast })>=0,  \label{eq1.16b}
\end{eqnarray}%
owing to the condition \ (\ref{eq1.3a}). The Hamiltonian function \ (\ref%
{eq1.6}) is obtained from the first three invariants of \ (\ref{eq1.16a}) as
\begin{equation}
H=\frac{2}{\alpha h^{3}}\bar{H}_{0}+\frac{1}{h}(H_{0}+H_{0}^{\ast }).
\label{eq1.16c}
\end{equation}

\begin{remark}
Similarly to the limiting condition \ (\ref{eq1.8}), the same limiting
expression one obtains from the discrete invariant function
\begin{equation}
H^{(w)}=\frac{1}{2\alpha h^{3}}\bar{H}_{0}-\frac{\alpha h}{4}%
(H_{1}+H_{1}^{\ast }),  \label{eq1.16d}
\end{equation}%
that is
\begin{equation}
\lim_{h-0}H^{(w)}=\tilde{H}:=\frac{1}{2}\int\limits_{\mathbb{R}}dx\left[
\psi \psi _{xx}^{\ast }+\psi _{xx}\psi ^{\ast }-2\alpha (\psi ^{\ast }\psi
)^{2}\right] .  \label{eq1.16e}
\end{equation}
\end{remark}

Based on these results, one can apply to the discrete dynamical system (\ref%
{eq1.3}) the Bogoyavlensky-Novikov type reduction scheme, devised before in
\cite{No,PBPS} and obtain a completely Liouville integrable finite
dimensional dynamical system on the manifold $M_{h}^{(N)}.$ Namely, we
consider the critical submanifold $M_{h}^{(N)}\subset M_{h}$ of the
following real-valued action functional \ \
\begin{equation}
\mathcal{L}_{_{h}}^{(N)}:=\sum_{n\in \mathbb{Z}}\mathcal{L}%
_{_{h}}^{(N)}[\psi _{n},\psi _{n}^{\ast }]=\bar{c}_{0}(h)\bar{H}%
_{0}+\sum_{j=0}^{N\ }c_{j}(h)(H_{j}+H_{j}^{\ast }),  \label{eq1.17}
\end{equation}%
where, by definition, $\bar{c}_{0},c_{j}:\mathbb{R}_{+}\rightarrow \mathbb{R}%
,$ $j=\overline{0,N},$ are suitably defined functions for arbitrary but
fixed $N\in \mathbb{Z}_{+},$ and
\begin{equation}
M_{h}^{(N)}:=\left\{ (\psi ,\psi ^{\ast })^{\intercal }\in M_{h}:\ \mathrm{%
grad}\mathcal{L}_{h}^{(N)}[\psi _{n},\psi _{n}^{\ast }]=0,\ n\in \mathbb{Z}%
\right\} .  \label{eq1.18}
\end{equation}

As one can easily show, the submanifold $M_{h}^{(N)}\subset M_{h}$ is
finite-dimensional and for any $N\in \mathbb{Z}_{+}$ is invariant with
respect to the vector field $K:M_{h}\rightarrow T(M_{h}).$ This property
makes it possible to reduce it on the submanifold $M_{h}^{(N)}\subset M_{h}$
and to obtain a resulting finite-dimensional system of ordinary differential
equations on $M_{h}^{(N)},$ whose solution manifold coincides with an
subspace of exact solutions to the initial dynamical system (\ref{eq1.3}).
The latter proves to be canonically Hamiltonian on the manifold $M_{h}^{(N)}$
and, moreover, completely Liouville-Arnold integrable. \ If the mappings $%
\bar{c}_{0},c_{j}:\mathbb{R}_{+}\rightarrow \mathbb{R},$ $j=\overline{0,N},$
are chosen in such a way that the flow \ (\ref{eq1.3}), invariantly reduced
on the finite dimensional submanifold $M_{h}^{(N)}\subset M_{h},$ is
nonsingular as $h\rightarrow 0$ \ and complete, then the corresponding
solutions to the discrete dynamical system \ (\ref{eq1.3}) will respectively
approach those to the nonlinear integrable Schr\"odinger dynamical system \ (%
\ref{eq1.1}).

Below we will proceed to realizing this scheme for the most simple cases $%
N=1 $ and $N=2.$ Another way of analyzing the discrete dynamical system \ (%
\ref{eq1.3}), being interesting enough, consists in applying the approaches
recently devised in \cite{Ci,MQ} and based on the long-time behavior of the
chosen discretization subject to a fixed Hamiltonian function structure.

\section{\label{Sec_2}The finite dimensional reduction scheme: the case $N=1$%
}

Consider the following non degenerate action functional
\begin{equation*}
\mathcal{L}_{\ h}^{(1)}=\sum_{n\in \mathbb{Z}}\bar{c}_{0}(h)\ln |1-\alpha
h^{2}\psi _{n}^{\ast }\psi _{n}|+\ \sum_{n\in \mathbb{Z}}c_{0}(h)(\psi
_{n}^{\ast }\psi _{n+1}+\psi _{n}\psi _{n+1}^{\ast })+
\end{equation*}%
\begin{eqnarray}
&&+\sum_{n\in \mathbb{Z}}c_{1}(h)(\frac{1}{2}\psi _{n}^{2}\psi _{n-1}^{\ast
,2}+\psi _{n}\psi _{n+1}\psi _{n-1}^{\ast }\psi _{n}^{\ast }-\frac{\psi
_{n}\psi _{n-2}^{\ast }}{\alpha h^{2}}+  \label{eq2.1} \\
&&+\frac{1}{2}\psi _{n}^{2}\psi _{n+1}^{\ast ,2}+\psi _{n}\psi _{n+1}\psi
_{n+1}^{\ast }\psi _{n+2}^{\ast }-\frac{\psi _{n-1}\psi _{n+1}^{\ast }}{%
\alpha h^{2}})\   \notag
\end{eqnarray}%
with mappings $\ \bar{c}_{0},c_{j}:\mathbb{R}_{+}\rightarrow \mathbb{R},$ $j=%
\overline{0,1},$ taken as \
\begin{equation}
\bar{c}_{0}(h)=\frac{4\xi +1}{2\alpha h^{3}},\text{ \ }c_{0}(h)=\frac{\xi }{h%
},\text{ \ \ }c_{1}(h)=\frac{\alpha h}{4},  \label{eq2.1a}
\end{equation}%
and being easily determined for any $\xi \in \mathbb{R}$ from the condition
for existence of a limit as $h\rightarrow 0$:
\begin{equation}
\mathcal{\tilde{L}}_{\ }^{(1)}:=\lim_{h\rightarrow 0}\mathcal{L}_{\ h}^{(1)}.
\label{eq2.1b}
\end{equation}%
The corresponding invariant critical submanifold
\begin{equation}
M_{h}^{(1)}:=\left\{ (\psi ,\psi ^{\ast })^{\intercal }\in M_{h}:\ \mathrm{%
grad}\mathcal{L}_{h}^{(1)}[\psi _{n},\psi _{n}^{\ast }]=0,\ n\in \mathbb{Z}%
\right\}  \label{eq2.2}
\end{equation}%
is equivalent to the following system of discrete up-recurrent relationships
with respect to indices $n\in \mathbb{Z}:$%
\begin{eqnarray*}
\psi _{n+2} &=&-\frac{-\bar{c}_{0}(h)/c_{1}(h)\psi _{n}}{(\frac{1}{\alpha
h^{2}}-\psi _{n+1\ }\psi _{n+1}^{\ast })(\frac{1}{\alpha h^{2}}-\psi
_{n}\psi _{n}^{\ast })}+ \\
&&+ \frac{2\psi _{n-1}c_{0}(h)/c_{1}(h)+\psi _{n\ }(\psi _{n+1}\psi
_{n-1}^{\ast }+\psi _{n-1}\psi _{n+1}^{\ast })}{(\frac{1}{\alpha h^{2}}-\psi
_{n+1\ }\psi _{n+1}^{\ast })}+ \\
& & + \frac{(\psi _{n+1}^{2}+\psi _{n-1}^{2})\psi _{n\ }^{\ast }-\psi _{n-2}(%
\frac{1}{\alpha h^{2}}-\psi _{n-1\ }\psi _{n-1}^{\ast })}{(\frac{1}{\alpha
h^{2}}-\psi _{n+1\ }\psi _{n+1}^{\ast })}:= \\
&=&\Phi _{+}(\psi _{n+1},\psi _{n+1}^{\ast };\psi _{n},\psi _{n}^{\ast
};\psi _{n-1},\psi _{n-1}^{\ast }),
\end{eqnarray*}%
\begin{eqnarray}
\psi _{n+2}^{\ast } &=&-\frac{-\bar{c}_{0}(h)/c_{1}(h)\psi _{n}^{\ast }}{(%
\frac{1}{\alpha h^{2}}-\psi _{n+1\ }\psi _{n+1}^{\ast })(\frac{1}{\alpha
h^{2}}-\psi _{n}\psi _{n}^{\ast })}+  \label{eq2.3} \\
&&+ \frac{2\psi _{n-1}^{\ast }c_{0}(h)/c_{1}(h)+\psi _{n\ }^{\ast }(\psi
_{n+1}\psi _{n-1}^{\ast }+\psi _{n-1}\psi _{n+1}^{\ast })}{(\frac{1}{\alpha
h^{2}}-\psi _{n+1\ }\psi _{n+1}^{\ast })}+  \notag \\
&& + \frac{(\psi _{n+1}^{\ast ,2}+\psi _{n-1}^{\ast ,2})\psi _{n\ }-\psi
_{n-2}^{\ast }(\frac{1}{\alpha h^{2}}-\psi _{n-1\ }\psi _{n-1}^{\ast })}{(%
\frac{1}{\alpha h^{2}}-\psi _{n+1\ }\psi _{n+1}^{\ast })}:=  \notag \\
&=&\Phi _{+}^{\ast }(\psi _{n+1},\psi _{n+1}^{\ast };\psi _{n},\psi
_{n}^{\ast };\psi _{n-1},\psi _{n-1}^{\ast }).  \notag
\end{eqnarray}%
The latter can be also rewritten as the system of down-recurrent mappings
\begin{eqnarray}
\psi _{n-2} &=&-\frac{-\bar{c}_{0}(h)/c_{1}(h)\psi _{n}}{(\frac{1}{\alpha
h^{2}}-\psi _{n-1\ }\psi _{n-1}^{\ast })(\frac{1}{\alpha h^{2}}-\psi
_{n}\psi _{n}^{\ast })}+  \label{eq2.4} \\
&&+\frac{2\psi _{n-1}c_{0}(h)/c_{1}(h)+\psi _{n\ }(\psi _{n+1}\psi
_{n-1}^{\ast }+\psi _{n-1}\psi _{n+1}^{\ast })}{(\frac{1}{\alpha h^{2}}-\psi
_{n-1\ }\psi _{n-1}^{\ast })}+  \notag \\
&&+\frac{(\psi _{n+1}^{2}+\psi _{n-1}^{2})\psi _{n\ }^{\ast }-\psi _{n+2}(%
\frac{1}{\alpha h^{2}}-\psi _{n+1\ }\psi _{n+1}^{\ast })}{(\frac{1}{\alpha
h^{2}}-\psi _{n-1\ }\psi _{n-1}^{\ast })}:=  \notag \\
&=&\Phi _{-}(\psi _{n-1},\psi _{n-1}^{\ast };\psi _{n},\psi _{n}^{\ast
};\psi _{n+1},\psi _{n+1}^{\ast }),  \notag
\end{eqnarray}%
\begin{eqnarray*}
\psi _{n-2}^{\ast } &=&-\frac{-\bar{c}_{0}(h)/c_{1}(h)\psi _{n}^{\ast }}{(%
\frac{1}{\alpha h^{2}}-\psi _{n-1\ }\psi _{n-1}^{\ast })(\frac{1}{\alpha
h^{2}}-\psi _{n}\psi _{n}^{\ast })}+ \\
&&+\frac{2\psi _{n-1}^{\ast }c_{0}(h)/c_{1}(h)+\psi _{n\ }^{\ast }(\psi
_{n+1}\psi _{n-1}^{\ast }+\psi _{n-1}\psi _{n+1}^{\ast })}{(\frac{1}{\alpha
h^{2}}-\psi _{n-1\ }\psi _{n-1}^{\ast })}+ \\
&&+\frac{(\psi _{n+1}^{\ast ,2}+\psi _{n-1}^{\ast ,2})\psi _{n\ }-\psi
_{n+2}^{\ast }(\frac{1}{\alpha h^{2}}-\psi _{n+1\ }\psi _{n+1}^{\ast })}{(%
\frac{1}{\alpha h^{2}}-\psi _{n-1\ }\psi _{n-1}^{\ast })}:= \\
&=&\Phi _{-}^{\ast }(\psi _{n-1},\psi _{n-1}^{\ast };\psi _{n},\psi
_{n}^{\ast };\psi _{n+1},\psi _{n+1}^{\ast }),
\end{eqnarray*}%
which also hold for all $n\in \mathbb{Z}.$ The \ relationships \ (\ref{eq2.3}%
) (or, the same, relationships\ (\ref{eq2.4})) mean that the whole
submanifold $M_{h}^{(1)}\subset M_{h}$ \ is retrieved by means of the
initial values \ $(\bar{\psi}_{-1},\bar{\psi}_{-1}^{\ast };\bar{\psi}_{0},%
\bar{\psi}_{0}^{\ast };\bar{\psi}_{1},\bar{\psi}_{1}^{\ast };\bar{\psi}_{2},%
\bar{\psi}_{2}^{\ast })^{\intercal }\in M_{h}^{(1)}\simeq \mathbb{C}^{8}.$
Thereby, the submanifold $M_{h}^{(1)}\subset M_{h}^{8}$ is naturally
diffeomorphic to the finite dimensional complex space $M_{h}^{8}.$ Taking
into account the canonical symplecticity \cite{Pr,PBPS} of the submanifold $%
M_{h}^{(1)}\simeq M_{h}^{8}$ and its invariance with respect to the vector
field \ (\ref{eq1.3}) one can easily reduce it on this submanifold $%
M_{h}^{(1)}\simeq M_{h}^{8}$ and obtain the following equivalent finite
dimensional flow on the manifold $M_{h}^{8}:$%
\begin{eqnarray*}
\frac{d}{dt}\psi _{2}\ &=&\frac{i}{h^{2}}[\Phi _{+}(\psi _{\ 2},\psi _{\
2}^{\ast };\psi _{1},\psi _{1}^{\ast };\psi _{\ 0},\psi _{\ 0}^{\ast
})-2\psi _{2}+\psi _{1}]- \\
&&-i\alpha \lbrack \Phi _{+}(\psi _{\ 2},\psi _{\ 2}^{\ast };\psi _{1},\psi
_{1}^{\ast };\psi _{\ 0},\psi _{\ 0}^{\ast })+\psi _{1}]\psi _{2}\psi
_{2}^{\ast },\
\end{eqnarray*}

\begin{equation}
\frac{d}{dt}\psi _{1}\ =\frac{i}{h^{2}}[\psi _{2}-2\psi _{1}+\psi
_{0}]-i\alpha (\psi _{2}+\psi _{0})\psi _{1}\psi _{1}^{\ast },\   \notag
\end{equation}

\bigskip
\begin{equation}
\frac{d}{dt}\psi _{0}\ =\frac{i}{h^{2}}[\psi _{1}-2\psi _{0}+\psi
_{-1}]-i\alpha (\psi _{1}+\psi _{-1})\psi _{0}\psi _{0}^{\ast },
\label{eq2.5}
\end{equation}

\bigskip
\begin{eqnarray}
\frac{d}{dt}\psi _{-1}\ &=&\frac{i}{h^{2}}[\psi _{0}-2\psi _{-1}+\Phi
_{-}(\psi _{\ -1},\psi _{\ -1}^{\ast };\psi _{0},\psi _{0}^{\ast };\psi _{\
1},\psi _{\ 1}^{\ast })]-  \notag \\
&&-\ i\alpha \lbrack \psi _{0}+\Phi _{-}(\psi _{\ -1},\psi _{\ -1}^{\ast
};\psi _{0},\psi _{0}^{\ast };\psi _{\ 1},\psi _{\ 1}^{\ast })]\psi
_{-1}\psi _{-1}^{\ast },  \notag
\end{eqnarray}%
and

\bigskip

\begin{eqnarray}
\frac{d}{dt}\psi _{-1}^{\ast }\ &=&-\frac{i}{h^{2}}[\psi _{0}^{\ast }-2\psi
_{-1}^{\ast }+\Phi _{-}^{\ast }(\psi _{\ -1},\psi _{\ -1}^{\ast };\psi
_{0},\psi _{0}^{\ast };\psi _{\ 1},\psi _{\ 1}^{\ast })]+  \notag \\
&&+\ i\alpha \lbrack \psi _{0}^{\ast }+\Phi _{-}^{\ast }(\psi _{\ -1},\psi
_{\ -1}^{\ast };\psi _{0},\psi _{0}^{\ast };\psi _{\ 1},\psi _{\ 1}^{\ast
})]\psi _{-1}\psi _{-1}^{\ast },  \notag
\end{eqnarray}%
\begin{equation*}
\frac{d}{dt}\psi _{0}^{\ast }\ =-\frac{i}{h^{2}}[\psi _{1}^{\ast }-2\psi
_{0}^{\ast }+\psi _{-1}^{\ast }]+i\alpha (\psi _{1}^{\ast }+\psi _{-1}^{\ast
})\psi _{0}\psi _{0}^{\ast },\
\end{equation*}%
\begin{equation}
\frac{d}{dt}\psi _{1}^{\ast }\ =-\frac{i}{h^{2}}[\psi _{2}^{\ast }-2\psi
_{1}^{\ast }+\psi _{0}^{\ast }]+i\alpha (\psi _{2}^{\ast }+\psi _{0}^{\ast
})\psi _{1}\psi _{1}^{\ast },\   \label{eq2.6}
\end{equation}%
\begin{eqnarray*}
\frac{d}{dt}\psi _{2}^{\ast }\ &=&-\frac{i}{h^{2}}[\Phi _{+}^{\ast }(\psi
_{\ 2},\psi _{\ 2}^{\ast };\psi _{1},\psi _{1}^{\ast };\psi _{\ 0},\psi _{\
0}^{\ast })-2\psi _{2}^{\ast }+\psi _{1}^{\ast }]+ \\
&&+i\alpha \lbrack \Phi _{+}^{\ast }(\psi _{\ 2},\psi _{\ 2}^{\ast };\psi
_{1},\psi _{1}^{\ast };\psi _{\ 0},\psi _{\ 0}^{\ast })+\psi _{1}^{\ast
}]\psi _{2}\psi _{2}^{\ast }.\
\end{eqnarray*}%
The next proposition, characterizing the Hamiltonian structure of the
reduced dynamical system \ (\ref{eq2.5}) and \ (\ref{eq2.6}), holds.

\begin{proposition}
The \ eight-dimensional complex dynamical system \ (\ref{eq2.5}) and (\ref%
{eq2.6}) is Hamiltonian on the manifold $M_{h}^{(1)}\simeq M_{h}^{8}$ \ with
respect to the canonical symplectic structure
\begin{equation}
\omega _{h}^{(2)}=\sum_{j=\overline{-2,1}}(dp_{-j}\wedge d\psi
_{-j}+dp_{-j}^{\ast }\wedge d\psi _{-j}^{\ast }),  \label{eq2.7}
\end{equation}%
where, by definition,
\begin{equation}
p_{-j}:=\mathcal{L}_{h,\psi _{n-j+1}}^{(1)\prime ,\ast }[\psi _{n},\psi
_{n}^{\ast }]\cdot 1,\text{ \ \ \ }p_{-j}^{\ast }:=\mathcal{L}_{h,\psi
_{n-j+1}^{\ast }}^{(1)\prime ,\ast }[\psi _{n},\psi _{n}^{\ast }]\cdot 1
\label{eq2.8}
\end{equation}%
for $j=\overline{-2,1}\ $ modulo the constraint $\mathrm{grad}\mathcal{L}%
_{h}^{(1)}[\psi _{n},\psi _{n}^{\ast }]=0,n\in \mathbb{Z},$ on the
submanifold $M_{h}^{(1)}\simeq M_{h}^{8},$ and the sign $^{"\prime ,\ast "}$
means the corresponding discrete Frech\`{e}t up-directed derivative and its
natural conjugation with respect to the convolution mapping on $T^{\ast
}(M_{h}^{(1)})\times T(M_{h}^{(1)}).$
\end{proposition}

\begin{proof}
The symplectic structure \ (\ref{eq2.7}) easily follows \cite{Pr,PBPS,BPS}
from the discrete version of the Gelfand-Dikiy \cite{GD} differential
relationship:%
\begin{eqnarray}
d\mathcal{L}_{h}^{(1)}[\psi _{n},\psi _{n}^{\ast }] &=&<\mathrm{grad}%
\mathcal{L}_{h}^{(1)}[\psi _{n-1},\psi _{n-1}^{\ast }],(d\psi _{n-1},d\psi
_{n-1}^{\ast })^{\intercal }>+  \label{eq2.9} \\
&&+\frac{d}{dn}\alpha _{h}^{(1)}(\psi _{n-1},\psi _{n-1}^{\ast };\psi
_{n},\psi _{n}^{\ast };\psi _{n+1},\psi _{n+1}^{\ast };\psi _{n+2},\psi
_{n+2}^{\ast }),  \notag
\end{eqnarray}%
where $\alpha _{h}^{(1)}\in \Lambda ^{1}(M_{h}^{(1)})$ is, owing to the
condition $\mathrm{grad}\mathcal{L}_{h}^{(1)}[\psi _{n}\ ,\psi _{n}^{\ast
}]=0,n\in \mathbb{Z},$ on the submanifold $M_{h}^{(1)},$ not depending on
the index $n\in \mathbb{Z}\ $ and suitably defined one-form on the manifold $%
M_{h}^{8},$ allowing the following canonical representation:%
\begin{equation}
\alpha _{h}^{(1)}=\sum_{j=\overline{-2,1}}(p_{-j}(h)d\psi _{-j}+p_{-j}^{\ast
}(h)d\psi _{-j}^{\ast })  \label{eq2.10}
\end{equation}%
with functions $p_{-j},p_{-j}^{\ast }:M_{h}^{(1)}\times \mathbb{R}%
\rightarrow \mathbb{C},$ $\ j=\overline{-2,1}.$ The latter, being defined by
the expressions \ (\ref{eq2.8}), compile jointly with variables $\psi
_{-j},\psi _{-j}^{\ast }:M_{h}^{(1)}\times \mathbb{R}\rightarrow \mathbb{C},$
$\ j=\overline{-2,1},$ the global coordinates on the finite dimensional
symplectic manifold $M_{h}^{8},$ proving the proposition.
\end{proof}

The dynamical system \ (\ref{eq2.5}) and (\ref{eq2.6}) \ on the manifold $%
M_{h}^{8}$ possesses, except its Hamiltonian function, additionally exactly
four mutually commuting functionally independent conservation laws $\mathcal{%
H}_{k},\mathcal{H}_{k}^{\ast }:M_{h}^{8}\rightarrow \mathbb{R},$ $k=%
\overline{0,1},$ \ and one Casimir function $\mathcal{\bar{H}}%
_{0}:M_{h}^{8}\rightarrow \mathbb{R},$ which can be calculated \cite{PBPS}
from the following functional relationships
\begin{eqnarray}
&<&\mathrm{grad}H_{k}[\psi _{n},\psi _{n}^{\ast }],K[\psi _{n},\psi
_{n}^{\ast }]>:=  \label{eq2.11} \\
& =& -\frac{d}{dn}\mathcal{H}_{k}(\psi _{n-1},\psi _{n-1}^{\ast };\psi
_{n},\psi _{n}^{\ast };\psi _{n+1},\psi _{n+1}^{\ast };\psi _{n+2},\psi
_{n+2}^{\ast }),  \notag
\end{eqnarray}

\begin{eqnarray}
&<&\mathrm{grad}H_{k}^{\ast }[\psi _{n},\psi _{n}^{\ast }],K[\psi _{n},\psi
_{n}^{\ast }]>:=  \notag \\
&=&-\frac{d}{dn}\mathcal{H}_{k}^{\ast }(\psi _{n-1},\psi _{n-1}^{\ast };\psi
_{n},\psi _{n}^{\ast };\psi _{n+1},\psi _{n+1}^{\ast };\psi _{n+2},\psi
_{n+2}^{\ast }),  \notag
\end{eqnarray}

\begin{eqnarray}
&<&\mathrm{grad}\bar{H}_{0}[\psi _{n},\psi _{n}^{\ast }],K[\psi _{n},\psi
_{n}^{\ast }]>:=  \notag \\
&=&-\frac{d}{dn}\mathcal{\bar{H}}_{0}(\psi _{n-1},\psi _{n-1}^{\ast };\psi
_{n},\psi _{n}^{\ast };\psi _{n+1},\psi _{n+1}^{\ast };\psi _{n+2},\psi
_{n+2}^{\ast }),  \notag
\end{eqnarray}%
for $k=\overline{0,1}\ $modulo the constraint $\mathrm{grad}\mathcal{L}%
_{h}^{(1)}[\psi _{n-2},\psi _{n-2}^{\ast }]=0,n\in \mathbb{Z},$ on the
manifold $M_{h}^{(1)}\simeq M_{h}^{8},$ where $\frac{d}{dn}:=\Delta -1$ is a
discrete analog of the differentiation and the shift operator $\Delta $ acts
as $\Delta f_{n}:=f_{n+1},n\in \mathbb{Z},$ for any mapping $\ f:\mathbb{%
Z\rightarrow C}.\ $\ \ From \ (\ref{eq2.11}) one can obtain by menas of
simple but tedious calculations analytical expressions for the invariants $%
\mathcal{H}_{k}^{\ast }:M_{h}^{8}\rightarrow \mathbb{R},$ \ which give rise
to the corresponding Hamiltonian function for the dynamical system \ (\ref%
{eq2.5}) and (\ref{eq2.6}), owing to the relationship \ (\ref{eq1.16b}):
\begin{equation}
\mathcal{H}=\frac{2}{\alpha h^{3}}\mathcal{\bar{H}}_{0}+\frac{1}{h}(\mathcal{%
H}_{0}+\mathcal{H}_{0}^{\ast }),  \notag
\end{equation}
satisfying the following canonical Hamiltonian system with respect to the
symplectic structure \ (\ref{eq2.7}):%
\begin{eqnarray}
d\psi _{-j}/dt &=&\partial \mathcal{H}/\partial p_{-j},\text{ \ \ \ }d\psi
_{-j}^{\ast }/dt=\partial \mathcal{H}/\partial p_{-j}^{\ast },
\label{eq2.13} \\
dp_{-j}/dt &=&-\partial \mathcal{H}/\partial \psi _{-j},\text{ \ \ \ }%
dp_{-j}^{\ast }/dt=-\partial \mathcal{H}/\partial \psi _{-j}^{\ast },  \notag
\end{eqnarray}%
where $j=\overline{-2,1},$ which is a Liouville-Arnold integrable on the
symplectic manifold $M_{h}^{8}.$

\bigskip

\begin{remark}
The same way on can construct the finite dimensional reduction of the
discrete Schr\"{o}dinger dynamical system \ (\ref{eq1.3}) in the case $N=2.$
Making use of the calculated before conservation laws \ (\ref{eq1.16a}), one
can take the corresponding action functional as
\begin{eqnarray}  \label{eq3.1}
& \mathcal{L}_{h}^{(2)} & = \ \bar{c}_{0}(h) \sum_{n\in \mathbb{Z}} \ln
|1-\alpha h^{2}\psi _{n}^{\ast }\psi _{n}|+\ c_{0}(h) \sum_{n\in \mathbb{Z}}
(\psi _{n}^{\ast }\psi _{n-1}+\psi _{n}\psi _{n-1}^{\ast })+  \notag \\
&&+ \ c_{1}(h) \sum_{n\in \mathbb{Z}} \left( \frac{1}{2}\psi _{n}^{2} \psi
_{n-1}^{\ast,2} + \psi _{n} \psi _{n}^{\ast } ( \psi _{n+1}\psi _{n-1}^{\ast
} + \psi _{n-1}\psi _{n+1}^{\ast }) \right. \\
&& \left. + \ \frac{1}{2}\psi _{n-1}^{2}\psi _{n}^{\ast ,2} -\frac{%
\psi_{n}\psi _{n-2}^{\ast }}{\alpha h^{2}} -\frac{\psi _{n-2}\psi _{n}^{\ast
}}{\alpha h^{2}} \right) +  \notag \\
&&+ \ c_{2}(h) \sum_{n\in \mathbb{Z}} \left( \frac{1}{3}\psi _{n}^{3}\psi
_{n-1}^{\ast ,3}+\psi _{n}\psi _{n+1}\psi _{n-1}^{\ast }\psi _{n}^{\ast
}(\psi _{n}\psi _{n-1}^{\ast }+\psi _{n+1}\psi _{n}^{\ast }+ \right.  \notag
\\
&&+ \ \psi _{n+2}\psi _{n+1}^{\ast })-\frac{\psi _{n}\psi _{n-1}^{\ast }}{%
\alpha h^{2}}(\psi _{n}\psi _{n-2}^{\ast }+\psi _{n+1}\psi _{n-1}^{\ast })-
\notag \\[1ex]
&&- \ \frac{\psi _{n}\psi _{n}^{\ast }}{\alpha h^{2}}(\psi _{n+1}\psi
_{n-2}^{\ast }+\psi _{n+2}\psi _{n-1}^{\ast })+\frac{\ \psi _{n}\psi
_{n-3}^{\ast }}{\alpha ^{2}h^{4}}+\frac{1}{3}\psi _{n}^{\ast ,3}\psi
_{n-1}^{3}+  \notag \\[1ex]
&&+ \ \psi _{n}^{\ast }\psi _{n+1}^{\ast }\psi _{n-1}\psi _{n}(\psi
_{n}^{\ast }\psi _{n-1}+\psi _{n+1}^{\ast }\psi _{n}+\psi _{n+2}^{\ast }\psi
_{n+1})-  \notag \\[1ex]
&& \left. - \ \frac{\psi _{n}^{\ast }\psi _{n-1}}{\alpha h^{2}}(\psi
_{n}^{\ast }\psi _{n-2}+\psi _{n+1}^{\ast }\psi _{n-1})-\frac{\psi _{n}\psi
_{n}^{\ast }}{\alpha h^{2}}(\psi _{n+1}^{\ast }\psi _{n-2}+\psi _{n+2}^{\ast
}\psi _{n-1})+\frac{\ \psi _{n}^{\ast }\psi _{n-3}}{\alpha ^{2}h^{4}} \right)
\notag
\end{eqnarray}%
with mappings $\bar{c}_{0},\ c_{j}:\mathbb{R}_{+}\rightarrow \mathbb{R},$ $j=%
\overline{0,2},$ defined as before from the condition that there exists the
limit
\begin{equation}
\mathcal{\tilde{L}}^{(2)}:=\lim_{h\rightarrow 0}\mathcal{L}_{h}^{(2)}.
\label{eq3.1a}
\end{equation}%
The respectively defined critical submanifold
\begin{equation}
M_{h}^{(2)}:=\left\{ (\psi ,\psi ^{\ast })^{\intercal }\in M_{h}:\ \mathrm{%
grad}\mathcal{L}_{h}^{(2)}[\psi _{n},\psi _{n}^{\ast }]=0,\ n\in \mathbb{Z}%
\right\}  \label{eq3.2}
\end{equation}%
becomes diffeomorphic to a finite dimensional canonically symplectic
manifold $M_{h}^{12}$ on which the suitably reduced discrete Schr\"odinger
dynamical system \ (\ref{eq1.3}) becomes a Liouville-Arnold integrable
Hamiltonian system. The details of the related calculations are planned to
be presented in a separate work under preparation.
\end{remark}

\section{\label{Sec_3} The Fourier analysis of the integrable discrete NLS
system}

It easy to observe that the linearized Schr\"{o}dinger system (\ref{eq1.1})
admits the following Fourier type solution:
\begin{equation}
\psi (x,t)=\int_{\mathbb{R}}ds\ \xi (s,t)\exp (ixs),\text{ \ \ \ }\psi
^{\ast }(x,t)=\int_{\mathbb{R}}ds\ \xi ^{\ast }(s,t)\exp (-ixs)\
\label{psi-F}
\end{equation}%
for all $x,t\in \mathbb{R}$, where $d\xi /dt=-is^{2}\xi ,$ $d\xi ^{\ast
}/dt=\ is^{2}\xi ^{\ast },$i.e.,
\begin{equation}
\xi (s,t)=\bar{\xi}(s)e^{-is^{2}t},\xi ^{\ast }(s,t)=\bar{\xi}^{\ast
}(s)e^{is^{2}t}
\end{equation}%
and $\bar{\xi},\bar{\xi}^{\ast }:\mathbb{R}\rightarrow \mathbb{C}$ are
prescribed functions (the Fourier transforms of   initial conditions).
Likewise, the linearized discrete Schr\"{o}dinger dynamical system \ (\ref%
{eq1.3}) allows the following general discrete Fourier\ type solution:%
\begin{equation}
\psi _{n}=\int_{\mathbb{R}}ds\ \xi _{h}(s,t)\exp (ihns),\text{ \ \ \ }\psi
_{n}^{\ast }=\int_{\mathbb{R}}ds\ \xi _{h}^{\ast }(s,t)\exp (-ihns)\
\label{eq4.1}
\end{equation}%
for all $n\in \mathbb{Z},$ where the evolution parameter $t\in \mathbb{R}%
,(\psi _{n},\psi _{n}^{\ast })^{\intercal }\in w_{h,2}^{2}(\mathbb{Z};%
\mathbb{C}^{2})$ and
\begin{equation}
\xi _{h}(s,t)=\bar{\xi}_{h}(s)\exp (-i\frac{4t}{h^{2}}\sin ^{2}\frac{sh}{2}%
),\ \qquad \xi _{h}^{\ast }(s,t)=\bar{\xi}_{h}^{\ast }(s)\exp (\ i\frac{4t}{%
h^{2}}\sin ^{2}\frac{sh}{2}).  \label{eq4.2}
\end{equation}%
Here the function $\ (\bar{\xi}_{h},\bar{\xi}_{h}^{\ast })^{\intercal }\in $
$W_{h,2}^{2}(\mathbb{R};\mathbb{C}^{2})\subset L_{2}(\mathbb{R};\mathbb{C}%
^{2}),$ where the functional space $W_{h,2}^{2}(\mathbb{R};\mathbb{C}^{2})$
is yet to be determined. From the boundary condition $(\psi _{n},\psi
_{n}^{\ast })^{\intercal }\in w_{h,2}^{2}(\mathbb{Z};\mathbb{C}^{2})$ it
follows that expressions
\begin{eqnarray}
\frac{1}{2\pi }\sum_{n\in \mathbb{Z}}\psi _{n}^{\ast }\psi _{n} &=&\int_{%
\mathbb{R}}ds\xi _{h}^{\ast }(s)\xi _{h}(s)<\infty ,  \label{eq4.3} \\
\frac{1}{2\pi }\sum_{n\in \mathbb{Z}}(\psi _{n+1}^{\ast }\psi _{n}+\psi _{n\
}^{\ast }\psi _{n+1}) &=&2\int_{\mathbb{R}}ds\cos (hs)\xi _{h}^{\ast }(s)\xi
_{h}(s)<\infty ,  \notag
\end{eqnarray}%
ensure the boundedness of the Hamiltonian function \ (\ref{eq1.6}), thereby
determining a functional space $W_{h,2}^{2}(\mathbb{R};\mathbb{C}^{2})$ to
which belong the vector function $(\xi _{h},\xi _{h}^{\ast })^{\intercal
}\in $ $L_{2}(\mathbb{R};\mathbb{C}^{2}).$ However the discrete evolution is
not following along the continuous trajectory.

Being motivated by works \cite{Ci-oscyl,CR}, we modify the discrete system
as folows in order to obtain the exact discretization:
\begin{equation}
\begin{array}{l}
\label{del-dNLS}\displaystyle\frac{d}{dt}\psi _{n}=\frac{i}{\delta ^{2}}%
(\psi _{n+1}-2\psi _{n}+\psi _{n-1})-i\alpha (\psi _{n+1}+\psi _{n-1})\psi
_{n}\psi _{n}^{\ast }, \\[2ex]
\displaystyle\frac{d}{dt}\psi _{n}^{\ast }=-\frac{i}{\delta ^{2}}(\psi
_{n+1}^{\ast }-2\psi _{n}^{\ast }+\psi _{n-1}^{\ast })+i\alpha (\psi
_{n+1}^{\ast }+\psi _{n-1}^{\ast })\psi _{n}\psi _{n}^{\ast },%
\end{array}%
\end{equation}%
Substituting (\ref{eq4.1}) into the linearization of (\ref{del-dNLS}) we
obtain
\begin{equation}
\xi _{h}(s,t)=\bar{\xi}_{h}(s)\exp (-i\frac{4t}{\delta ^{2}}\sin ^{2}\frac{sh%
}{2}),\ \qquad \xi _{h}^{\ast }(s,t)=\bar{\xi}_{h}^{\ast }(s)\exp (\ i\frac{%
4t}{\delta ^{2}}\sin ^{2}\frac{sh}{2}).  \label{xi-del}
\end{equation}%
Therefore,  linearization of the discretization (\ref{del-dNLS}) is exact
(i.e., $\psi (nh,t)=\psi _{n}(t),n\in \mathbb{Z},$ if we assume
\begin{equation}
\delta =\frac{2}{s}\sin \frac{hs}{2}\   \label{xi-del-1}
\end{equation}%
for any $h\in \mathbb{R}.$ $\ \ $Thus, the parameter $\delta >0$ depends on
$s\in \mathbb{R}\ $ yet for small $h\rightarrow 0$ one gets  $\delta
=h(1+O(h^{2}s^{2}).$

The nonlinear case is more difficult. In the continuous nonlinear case (\ref%
{psi-F}) we have
\begin{equation}
d\xi /dt=-is^{2}\xi -2i\alpha \beta \lbrack s;\xi ]\ ,\qquad \quad d\xi
^{\ast }/dt=is^{2}\xi ^{\ast }+2i\alpha \beta ^{\ast }[s;\xi ^{\ast }],
\label{xi-del-2}
\end{equation}
where  the functionals  $\ \beta ,\beta ^{\ast }:\mathbb{R\times }L_{2}(%
\mathbb{R};\mathbb{C})\ \rightarrow L_{2}(\mathbb{R};\mathbb{C}),$
determined as
\begin{eqnarray*}
\beta \lbrack s;\xi ] &:&=\int\limits_{\mathbb{R}^{2}}ds^{\prime } ds^{\prime\prime}\xi
(s+s^{\prime }-s^{\prime\prime})\xi (s^{\prime\prime})\xi ^{\ast }(s^{\prime }), \\
\beta ^{\ast }[s;\xi ] &:&=\int\limits_{\mathbb{R}^{2}}ds^{\prime }ds^{\prime\prime} \xi
^{\ast }(s+s^{\prime }-s^{\prime\prime})\xi ^{\ast }(s^{\prime\prime})\xi (s^{\prime }),
\end{eqnarray*}
depend  both on $s\in \mathbb{R}$ and on the element $\xi \in L_{2}(\mathbb{R%
};\mathbb{C}),$ as well as depends parametrically on the evolution parameter
$t\in \mathbb{R}$ through (\ref{xi-del-2}). In the nonlinear discrete case (%
\ref{del-dNLS})  we, respectively, obtain:
\begin{equation}
d\xi _{h}/dt=-i\xi _{h}\frac{4}{\delta ^{2}}\sin ^{2}\frac{sh}{2} - 2i\alpha
\beta _{h}[s;\xi _{h}],   \qquad      d\xi _{h}^{\ast }/dt=i\xi _{h}^{\ast }\frac{4}{%
\delta ^{2}}\sin ^{2}\frac{sh}{2}   +2i\alpha
\beta _{h}^{\ast }[s;\xi _{h}],   \\[2ex]    \label{xi-del-3}
\end{equation}%
where the functionals $\ \beta _{h},\beta _{h}^{\ast }:\mathbb{R\times }%
L_{2}(\mathbb{R};\mathbb{C})\ \rightarrow L_{2}(\mathbb{R};\mathbb{C})\ $
are determined as
\begin{eqnarray}
\beta _{h}[s;\xi _{h}] &:&=\int\limits_{\mathbb{R}^{2}}ds^{\prime }ds^{\prime\prime}\cos
[h(s+s^{\prime }-s^{\prime\prime})]\xi _{h}(s+s^{\prime }-s^{\prime\prime})\xi _{h}(s^{\prime\prime})\xi _{h}^{\ast
}(s^{\prime }),  \label{xi-del-4} \\
\beta _{h}^{\ast }[s;\xi _{h}] &:&=\int\limits_{\mathbb{R}^{2}}ds^{\prime
}ds^{\prime\prime}\cos [h(s+s^{\prime }-s^{\prime\prime})]\xi _{h}^{\ast }(s+s^{\prime }-s^{\prime\prime})\xi
_{h}^{\ast }(s^{\prime\prime})\xi _{h}(s^{\prime })     \notag
\end{eqnarray}%
for  any $s\in \mathbb{R}.$ To proceed further with the truly nonlinear case
still presist to be  a nontrivial problem, yet we hope to obtain a suitable
procedure analogous to that of \cite{Ci,CR-PRE}.

Instead of it one can analyze the related functional space \ constraints on
the space of functions $(\bar{\xi}_{h},\bar{\xi}_{h}^{\ast })^{\intercal
}\in $ $W_{h,2}^{2}(\mathbb{R};\mathbb{C}^{2}),$ representing solutions to
the discrete nonlinear equation (\ref{eq1.3}) via the expressions \ (\ref%
{eq4.1}), being imposed by the corresponding finite dimensional reduction
scheme of Section \ (\ref{Sec_2}). This procedure actually may be realized,
if to consider the derived before recurrence relationships \ (\ref{eq2.3})
(or similarly, \ (\ref{eq2.4})) allowing to obtain the related constraints
on the space of functions $(\bar{\xi}_{h},\bar{\xi}_{h}^{\ast })^{\intercal
}\in $ $W_{h,2}^{2}(\mathbb{R};\mathbb{C}^{2}),$ but the resulting
relationships prove  to be much complicated and cumbersome expressions.

Thus, one can suggest the following practical numerical-analytical scheme of
constructing solutions to the discrete nonlinear Schr\"{o}dinger dynamical
system \ (\ref{eq1.3}): \ first to solve the Cauchy problem to the
finite-dimensional system of ordinary differential equations \ (\ref{eq2.5})
and \ (\ref{eq2.6}), and next to substitute them into the system of
recurrent algebraic relationships \ (\ref{eq2.3}) and \ (\ref{eq2.4}),
obtaining this way the whole infinite hierarchy of the sought for solutions.

\section{Conclusion}

Within the presented investigation of solutions to the discrete nonlinear
Schr\"odinger dynamical system (\ref{eq1.3}) we have succeeded in two
important points. First, we have developed an effective enough scheme of
invariant reducing the infinite system of ordinary differential equations (%
\ref{eq1.3}) to an equivalent finite one of ordinary differential equations
with respect to the evolution parameter $t\in \mathbb{R}$. Second, we
constructed a finite set of recurrent algebraic regular relationships,
allowing to expand the obtained before solutions to any discrete order $n\in
\mathbb{Z}\ $and giving rise to the sought for solutions of the system (\ref%
{eq1.3}).$\ $

It is important to mention here that within the presented analysis we have
not used the Lax type representation for the discrete nonlinear
Schr\"odinger dynamical system \ (\ref{eq1.3}), whose existence was stated
many years ago in \cite{AL} and whose complete solution set analysis was
done in works \cite{AL,AL1,BP,No} by means of both the inverse scattering
transform and the algebraic-geometric methods. Concerning the\ set of
recurrent relationships for exact solutions to the finite-dimensional
reduction of the discrete nonlinear Schr\"odinger dynamical system \ (\ref%
{eq1.3}), obtained both in the presented work and in work \cite{BP}, based
on the corresponding Lax type representation, an interesting problem of
finding between them relationship arises, and an answer to it would explain
the hidden structure of the complete Liouville-Arnold integrability of the
related set of the reduced ordinary differential equations.

\section{Acknowledgements}

J.L.C. acknowledges a financial support from the National Science Center
(Poland) under NCN grant No. 2011/01/B/ST1/05137 and appreciates fruitful
discussions during the 6th Symposium on Integrable Systems held in Bia\l %
ystok (Poland) on 26-29 June, 2013. J.L.C. is also grateful for discussions and  
hospitality during his visit in Lviv (Ukraine), July, 2013. \ A.P. is thankful to Prof. D. Blackmore
for valuable discussions and remarks during the Nonlinear Mathematical
Physics Conference held in the Sophus Lie Center, Nordfjordeid (Norway) on
June 04-14, 2013.

\vspace{0.5cm}

\end{document}